\documentclass[conference,doublecolumn,a4paper]{IEEEtran}

\usepackage[english]{babel}
\usepackage[T1]{fontenc}
\usepackage[utf8]{inputenc}
\usepackage{subfig}
\usepackage{multirow}
\usepackage[font = scriptsize, labelfont = bf, labelsep = period, tableposition = top, singlelinecheck = false, justification = centering]{caption}
\usepackage{placeins,verbatim}
\usepackage{amssymb}
\usepackage{amsmath}
\usepackage{amsthm}
\usepackage{latexsym}
\usepackage{wrapfig}
\usepackage{float}
\usepackage{mathtools}
\usepackage{bm}
\usepackage{listings}
\usepackage{textcomp}
\usepackage{setspace}
\usepackage[dvipsnames]{xcolor}
\usepackage{etoolbox}
\usepackage{hyperref}
\usepackage{tikz}
\usetikzlibrary{shapes,snakes,calc}

\newtheorem{definition}{Definition}
\newtheorem{lemma}{Lemma}
\newtheorem{theorem}{Theorem}

\newtheorem{remark}{Remark}

\newif\ifcomment
\commenttrue
\newcommand{\ma}[1]{\ifcomment \textcolor{orange}{Matteo: #1} \fi}
\newcommand{\lh}[1]{\ifcomment \textcolor{OliveGreen}{Lukas: #1} \fi}

\newif\ifcommentLater
\commentLaterfalse
\newcommand{\todoLater}[1]{\ifcommentLater {#1} \fi}

\newcommand{\ie}{\emph{i.e.}, }

\definecolor{dartmouthgreen}{rgb}{0.05, 0.5, 0.06}

\makeatletter
\def\smalloverbrace#1{\mathop{\vbox{\m@th\ialign{##\crcr\noalign{\kern3\p@}%
  \tiny\downbracefill\crcr\noalign{\kern3\p@\nointerlineskip}%
  $\hfil\displaystyle{#1}\hfil$\crcr}}}\limits}
\makeatother

\makeatletter
\def\smallunderbrace#1{\mathop{\vtop{\m@th\ialign{##\crcr
   $\hfil\displaystyle{#1}\hfil$\crcr
   \noalign{\kern3\p@\nointerlineskip}%
   \tiny\upbracefill\crcr\noalign{\kern3\p@}}}}\limits}
\makeatother

\makeatletter
\renewcommand*\env@matrix[1][*\c@MaxMatrixCols c]{%
  \hskip -\arraycolsep
  \let\@ifnextchar\new@ifnextchar
  \array{#1}}
\makeatother

\AtBeginEnvironment{matrix}{\setlength{\arraycolsep}{5pt}}
\DeclareMathOperator{\diagonal}{diag}

\DeclareMathOperator{\dimension}{dim}

\DeclareMathOperator{\spacespan}{span}

\DeclareMathOperator{\lcm}{lcm}

\def\={\coloneqq}							% Equal by definition
\def\tpar#1{\left( #1 \right)}    			% ()
\def\gpar#1{\left\{ #1 \right\}} 			% {}
\def\diag#1{\diagonal\! \tpar{#1}}			% Diagonal
\def\dim#1{\dimension \tpar{#1}}			% Dimension
\def\Span#1{\spacespan\gpar{#1}}% Span
\def\cb#1{| #1 \rangle}						% Column bit
\def\rb#1{\langle #1 |}						% Raw bit
\def\ppmatrix#1{\begin{pmatrix}#1\end{pmatrix}} % Matrix
\def\q#1{\mathfrak{q}_\mathsf{#1}}

\def\a{\alpha}
\def\b{\beta}

\def\c{\gamma}
\def\C{\mathbb{C}}

\def\cB{\mathcal{B}}
\def\cC{\mathcal{C}}
\def\cD{\mathcal{D}}

\def\cJ{\mathcal{J}}
\def\cL{\mathcal{L}}

\def\cS{\mathcal{S}}
\def\cV{\mathcal{V}}
\def\cW{\mathcal{W}}
\def\co{\cb{1}}								% Column bit 1
\def\cz{\cb{0}}								% Column bit 0
\def\d{\delta}

\def\F{\mathbb{F}}

\def\Fq{\F_{q}}

\def\H{\mathcal{H}}							% Hilbert space

\def\k{\kappa}
\def\l{\lambda}

\def\N{\mathbb{N}}

\def\s{\sigma}
\def\smix{\s_\text{mix}}

\def\sW{\mathsf{W}}

\def\sX{\mathsf{X}}

\def\sZ{\mathsf{Z}}

\def\tV{\text{V}}

\newcounter{alp}
\newcounter{ara}
\newcounter{rom}

\usepackage[innermargin=0.545in,outermargin=0.545in,top=0.6in,bottom=0.75in]{geometry}
\title{High-Rate Quantum Private Information Retrieval with Weakly Self-Dual Star Product Codes}

\author{
  \IEEEauthorblockN{Matteo Allaix \IEEEauthorrefmark{1}, Lukas Holzbaur \IEEEauthorrefmark{2}, Tefjol Pllaha \IEEEauthorrefmark{1}, Camilla Hollanti \IEEEauthorrefmark{1}}

	\IEEEauthorblockA{\small \IEEEauthorrefmark{1} Aalto University, Finland. E-mails:
    \{matteo.allaix, tefjol.pllaha, camilla.hollanti\}@aalto.fi
	}
	\IEEEauthorblockA{\small \IEEEauthorrefmark{2} Technical University of Munich, Germany. E-mail:  lukas.holzbaur@tum.de
    }
  \thanks{L.~Holzbaur was supported by the German Research Foundation (Deutsche Forschungsgemeinschaft, DFG) under Grant No.~WA$3907/1-1$. 
  
  C.~Hollanti and M.~Allaix were supported by the Academy of Finland, under Grants No. 318937 and 336005. }
}
\IEEEoverridecommandlockouts

\begin{document}

\maketitle

\begin{abstract}
In the classical private information retrieval (PIR) setup, a user wants to retrieve a file from a database or a distributed storage system (DSS)  without revealing the file identity to the servers holding the data. In the quantum PIR (QPIR) setting, a user privately retrieves a classical file by receiving quantum information from the servers. The QPIR problem has been treated by Song \emph{et al.} in the case of replicated servers, both with and without collusion. QPIR over $[n,k]$ maximum distance separable (MDS) coded servers was recently considered by Allaix \emph{et al.}, but the collusion was essentially restricted to $t=n-k$ servers in the sense that a smaller $t$ would not improve the retrieval rate. In this paper, the QPIR setting is extended to allow for retrieval with high rate for any number of colluding servers $t$ with $1 \leq t \leq n-k$.
Similarly to the previous cases, the rates achieved are better than those known or conjectured in the classical counterparts, as well as those of the previously proposed coded and colluding QPIR schemes. This is enabled by considering the stabilizer formalism and weakly self-dual generalized Reed--Solomon (GRS) star product codes.
\end{abstract}

\section{Introduction}

Private information retrieval (PIR) \cite{chor1995private} enables a user to download a data item from a database without revealing the identity of the retrieved item to the database owner (user privacy). If additionally the user is supposed to obtain no information about any file other than the requested file (server privacy), the problem is referred to as \emph{symmetric} PIR (SPIR). In recent years, PIR has gained renewed interest in the setting of distributed storage systems (DSSs), where the servers are storing possibly large files and may \emph{collude}, \emph{i.e.}, exchange their obtained queries. To protect from data loss in the case of the failure of some number of servers, such systems commonly employ erasure-correcting codes, \emph{e.g.}, maximum distance separable (MDS) codes \cite{macwilliams1977theory}. 
%A further extension of the PIR problem considers the case of \emph{collusion}, where the index of the file requested by the user is required to be private even if an (unknown) subset of servers cooperates. 
The capacity of PIR is known in a variety of settings \cite{sun2017capacity,sun2018capacity,banawan2018capacity,wang2017linear,Banawan2019byzantine}, but is still open in its full generality for coded and colluding servers \cite{freij2017private,Sun2018conjecture}. Progress towards the general coded colluded PIR capacity was recently made in \cite{Holzbaur2019ITW,holzbaur2019capacity}.

\begin{table*}[t]
    \centering
    \caption{Known capacity results with $n$ servers. For the classical PIR capacities, we report the asymptotic results with respect to the number of files. %The acronym SPIR stands for symmetric PIR. 
    The result in red is a conjectured result, but a protocol achieving that rate was proposed in \cite{freij2017private}. The QPIR results in blue were proved for $n=2$,$n=t+1$, and $n=k+t$ servers, respectively, with 2-dimensional quantum systems. The other two QPIR results were proved with $q$-dimensional quantum systems. The result in green is proved in this paper.}
    \begin{tabular}{|l|cc|cc|cc|}
    \hline
    \textsc{Capacities} & PIR & ref. & SPIR & ref. & QPIR & ref. \\
    \hline
    Replicated storage, & \multirow{2}{*}{$1-\frac{1}{n}$} & \multirow{2}{*}{\cite{sun2017replicated}} & \multirow{2}{*}{$1-\frac{1}{n}$}  & \multirow{2}{*}{\cite{sun2018capacity}} & \multirow{2}{*}{\textcolor{blue}{1}} & \multirow{2}{*}{\cite{song2019capacitymultiple}} \\
    no collusion & & & & & & \\
    \hline
    Replicated storage, & \multirow{2}{*}{$1-\frac{t}{n}$} & \multirow{2}{*}{\cite{sun2017capacity}} & \multirow{2}{*}{$1-\frac{t}{n}$}  & \multirow{2}{*}{\cite{wang2017colluding}} & \textcolor{blue}{$\geq \frac{2}{t+2}$} & \cite{song2019capacitycollusion} \\
    $t$-collusion & & & & & $\min\gpar{1,\frac{2(n-t)}{n}}$ & \cite{song2020} \\
    \hline
    $[n,k]$-MDS coded & \multirow{2}{*}{\textcolor{red}{$1-\frac{k+t-1}{n}$}} & \multirow{2}{*}{\cite{freij2017private}} & \multirow{2}{*}{$1-\frac{k+t-1}{n}$}  & \multirow{2}{*}{\cite{wang2017linear}} & \textcolor{blue}{$\geq \frac{2}{k+t-1}$} & \cite{allaix2020quantum} \\
    storage, $t$-collusion & & & & & \textcolor{dartmouthgreen}{$\geq \min\gpar{1,\frac{2(n-k-t+1)}{n}}$} & -- \\
    \hline
    \end{tabular}
    \label{tab:Capacities}
\end{table*}
%\end{minipage}

The problem of PIR has also been considered in the quantum communication setting \cite{kerenidis2004quantum, gall2011quantum,giovannetti2008quantum}, where the problem is referred to as \emph{quantum} PIR (QPIR). More recently, Song \emph{et al.} \cite{song2019capacitymultiple,song2019capacitycollusion,song2020} introduced a scheme for a replicated storage system with classical files, where the servers respond to user's (classical) queries by sending quantum systems. The user is then able to privately retrieve the file by measuring the quantum systems. The servers are assumed to share some entangled states, while the user and the servers are not entangled. The non-colluding case was considered in \cite{song2019capacitymultiple}, and was shown to have capacity\footnote{The quantum PIR schemes in \cite{song2019capacitymultiple,song2019capacitycollusion,song2020} and in this work are symmetric. 
%\emph{i.e.}, the user obtains no information about any file except the requested one.
For the comparison of our rates to the classical setting we will focus on the asymptotic  non-symmetric rates for the latter, which also coincide with the SPIR rates, cf. Table \ref{tab:Capacities}.} equal to one. This is in stark contrast to the classical replicated (asymptotic) PIR  capacity of $1-\frac{1}{n}$ for $n$ servers. The case of QPIR for all but one servers colluding, \emph{i.e.}, $t=n-1$, was considered in \cite{song2019capacitycollusion}, again achieving higher capacity than the classical counterpart. In this case, the QPIR capacity is $\frac{2}{n}$, while classically (and asymptotically) it is $\frac{1}{n}$. This work was extended to $[n,k]$ MDS-coded data for collusion of up to  $t= n-k$ servers in \cite{allaix2020quantum}, and an analogous rate improvement was achieved. The scheme presented therein naturally also resists $t<n-k$ colluding servers. However, as pointed out in \cite[Remark 3]{allaix2020quantum}  %\lh{maybe add here: "
and in contrast to the scheme proposed in the present work, it is not able to benefit from the potential rate improvement made possible by the smaller collusion.
%\lh{I think the previous sentence is good, but the following sentence is maybe more confusing than helpful.}
%This improvement would correspond to the classical one obtained in \cite{freij2017private} generalizing the work in \cite{tajeddine2018private}, and is finally achieved in the present paper. 
% \ch{The scheme presented therein could naturally also resist against $t<n-k$ servers, however not being able to achieve a higher rate. As pointed out in Remark 3 in \cite{allaix2020quantum}, the potential rate improvement obtained by smaller collusion remained unharnessed. This improvement would correspond to the classical one obtained in \cite{freij2017private} generalizing the work in \cite{tajeddine2018private}, and is finally achieved in the present paper. (NOT SURE IF THIS IS ADDING TO THE CONFUSION RATHER THAN CLARIFYING, BUT I THINK IT STILL WASNT CLEAR WHAT THE IMPROVEMENT HERE IS.)}
In~\cite{song2020}, the authors extend their work \cite{song2019capacitymultiple,song2019capacitycollusion} by considering symmetric QPIR that can resist any $t$ servers colluding. They prove that the $t$-private QPIR capacity is $1$ for $1\leq t \leq n/2$ and $2(n-t)/n$ for $n/2 < t < n$ and they use the \emph{stabilizer formalism}~\cite{Gottesman97} to construct a capacity-achieving protocol. For the reader's convenience, we report some known results on the capacity in Table~\ref{tab:Capacities}.

%\subsection{Contributions}

\textbf{Contributions.} We consider a Generalized Reed--Solomon (GRS) coded storage system with (classical) files, where the servers respond to user's (classical) queries by sending quantum systems. The user is then able to privately retrieve the file by measuring the quantum systems. The servers are assumed to share some entangled state, while the user and the servers share no entanglement. We generalize the QPIR protocol for replicated storage systems protecting against collusion \cite{song2020} to the case of $[n,k]$-GRS coded servers and arbitrary $t$-collusion by applying the star product scheme~\cite{freij2017private}. Hence, the protocol of \cite{song2020} is the special case of $k=1$ in our protocol. This can be seen as trading off collusion protection for reduced storage overhead. The achieved rate $\sim \min\gpar{1,\frac{2(n-k-t+1)}{n}}$ (cf. Theorem~\ref{thm:rate}) is higher than the conjectured asymptotic rate $1 - \frac{k+t-1}{n}$ in the classical coded and colluding PIR~\cite{freij2017private}.
Note that this rate is $1$ if $n\geq 2(k+t-1)$, while in \cite{allaix2020quantum} the rate is $\sim \frac{2}{k+t}$ regardless of the number of servers.

\section{Basics on PIR and Quantum Computation} \label{sec:preliminaries}

%\subsection{Notation}
\textbf{Notation.} We denote by $[n]$ the set $\gpar{1,2,\ldots,n}, n \in \N$, and by $\Fq$ the finite field of $q$ elements. For a linear code of length $n$ and dimension $k$ over $\Fq$ we write $[n,k]$. For a matrix $A$ we write $A^\top$ for its transpose and $A^\dagger$ for its conjugate transpose. We will frequently deal with $m \b \times 2n$ matrices, where sub-blocks of $\beta$ rows and the pair of columns $s$ and $n+s$ semantically belong together. We therefore index such a matrix $Y$ by two pairs of indices $(i,b),\ i\in[m], b\in[\beta]$ and $(p,s),\ p\in [2], s\in [n]$, where $Y^{i,b}_{p,s}$ denotes the symbol in row $(i-1)\beta + b$ and column $(p-1)n+s$, \ie the symbol in the $b$-th row of the $i$-th sub-block of rows and the $s$-th column of the $p$-th sub-block of columns. Omitting of an index implies that we take all positions, \ie $Y^{i}$ denotes the $i$-th subblock of $\b$ columns, $Y^{i,b}$ the row $(i-1)\beta + b$, $Y_{p}$ the $p$-th subblock of $n$ columns, and $Y_{p,s}$ the column $(p-1)n+s$. For the reader's convenience, we sometimes imply the separation of the subblocks of columns by a vertical bar in the following. We denote by $e^\l_\c$ the standard basis column vector of length $\l$ in $\Fq^\l$ with a 1 in position $\c \in [\l]$. Given $a \in [\a], b \in [\b]$, it will help our notation to call \emph{coordinate} $(a,b)$ the position $\b (a - 1) + b$ in a vector of length $\a\b$. For instance, $e^{2\cdot 3}_{(2,1)}=e^6_4=(0,0,0,1,0,0)$. The function $\d_{i,j}$ is the Kronecker delta and $I_\nu$ is the $\nu \times \nu$ identity matrix. For a $\mu \times \nu$ zero matrix $0^{\mu \times \nu}$ and matrices $M_1,M_2 \in \Fq^{\mu \times \nu}$
\[
\diag{M_1,M_2} = \ppmatrix{M_1 & 0^{\mu \times \nu} \\ 0^{\mu \times \nu} & M_2} \in \Fq^{2\mu \times 2\nu}.
\]

For two vectors $c,d\in \F^n$ we define the (Hadamard-) star-product as $c\star d = (c_1d_1,c_2d_2,\ldots,c_nd_n)$. For two codes $\cC, \cD \subseteq \F^{n}$ we denote $\cC \star \cD = \langle \{ c \star d \ | \ c\in \cC, d\in \cD \} \rangle$.
Observe that, as the star-product is an element-wise operation, we have 
\begin{equation}\label{eq:starCartesian}
    \tpar{\cC \times \cC} \star \tpar{\cD \times \cD} = \tpar{\cC \star \cD} \times \tpar{\cC \star \cD} \ .
\end{equation}

\textbf{Linear codes and Distributed Data Storage.} We consider a distributed storage system employing error/erasure correcting codes to protect against data loss (for an illustration see Figure~\ref{fig:DSS}). 
To this end, let $X$ be an $m\beta \times 2k$ matrix containing $m$ files $X^i,\ i\in [m]$, each consisting of $2 \beta k$ symbols of $\Fq$. %\lh{Technically, we don't need the following, it's now defined as part of the notation.} The symbols of the $i$-th file are given by $X^i = 
%\begin{pmatrix}[c|c]
%    X_{1,\k}^{i,b} & X_{2,\k}^{i,b}
%\end{pmatrix}$,
%where $X_{p,\k}^{i,b}~\in~\Fq,\ p~\in~[2],\ \k~\in~[k],\ i~\in~[m],\ b~\in~[\b]$. %The $b$-th row of the $i$-th file $X^i$ is denoted by $X^{i,b}$. 

This matrix is encoded with a linear code $\cC$ of length $2n$ and dimension $2k$, which is the Cartesian product of an $[n,k]$ code over $\Fq$ with itself\footnote{We choose this description of the storage code because this structure is required for the quantum PIR scheme. However, note that the system can equivalently be viewed as being encoded with an $[n,k]$ over $\F_{{\sf q}^2}$, where each of the servers stores one column of the resulting codeword matrix.}, i.e., $\cC = \cC' \times \cC'$. It therefore has a generator matrix $G_\cC = \diag{G_{\cC'},G_{\cC'}}$, where $G_{\cC'}$ is a generator matrix of $\cC'$. The $m\beta \times 2n$ matrix of encoded files is given by $Y = X \cdot G_\cC$. Server $s \in [n]$ stores columns $s$ and $n+s$ of $Y$, \ie it stores $Y_{1,s}$ and $Y_{2,s}$. %, where the pair of symbols of the codewords corresponding to the $b$-th row of the $i$-th file $X^i$ is denoted by $Y^{i,b} =
%\begin{pmatrix}[c|c]
%    Y_{1,s}^{i,b} & Y_{2,s}^{i,b}
%\end{pmatrix}$.
%$Y^{i,b} = (Y_{1,s}^{i,b} \ | \ Y_{2,s}^{i,b})$.
% The resulting $m\beta \times 2n$ matrix of codewords is given by $y=x \cdot G_\cC$. The $s$-th server stores the columns with indices $s$ and $n+s$.

In this work we consider systems encoded with (the Cartesian product of) generalized Reed-Solomon (GRS) codes (cf.~\cite[Ch.~10]{macwilliams1977theory}), a popular class of MDS codes. Among coded storage systems, these have proven to be particularly well-suited for PIR and general schemes exist for a wide range of parameters \cite{tajeddine2018private, freij2017private, tajeddine2019private}. The key idea is to design the queries such that the retrieved symbols are the sum of a codeword of another GRS code (of higher dimension), which we refer to as the \emph{star-product code}, plus a vector depending only on the desired file. To obtain the desired file, the codeword part is projected to zero, leaving only desired part of the responses. In the QPIR system we consider in the following, this projection is part of the quantum measurement. This imposes a constraint on this star-product code, namely, that the code is (weakly) self-dual. In the following, we collect/establish the required theoretical results on GRS codes and their star-products.

\begin{definition}[Weakly self-dual code]
  \label{def:WeaklySelfDualCode}
  We say that an $[n,k]$ code $\cC$ is \emph{weakly self-dual} if $\cC^\perp \subseteq \cC$ and \emph{self-dual} if $\cC^\perp = \cC$. It is easy to see that any such code with parity-check matrix $H$ has a generator matrix of the form $G = (H^\top \ \ F^\top)^\top$ for some $(2k-n)\times n$ matrix $F$.
\end{definition}

\begin{lemma}[Follows from {\cite[Theorem~3]{grass2008self}}]\label{lem:existenceSelfDual}
  For $q=2^m$ there exist self-dual GRS $[2k,k]$ codes over $\Fq$ for any $k\in [2^{m-1}]$ and code locators $\cL$.
\end{lemma}

\begin{lemma}
  Let $q$ be even with $q\geq n$. Then there exists a weakly self-dual $[n,k]$ GRS code $\cC$ for any $k\geq \frac{n}{2}$ and code locators~$\cL$.
\end{lemma}
\begin{proof}
  Let $\cC_{[n,n/2]}$ be an $[n,n/2]$ self-dual GRS code with code locators $\mathcal{L}$, as shown to exist in \cite[Theorem~3]{grass2008self} (see Lemma~\ref{lem:existenceSelfDual}). It is easy to see that this code is a subcode of the $[n,k]$ GRS code $\cC_{[n,k]}$ with the same column multipliers. The property $\cC_{[n,k]}^\perp \subset \cC_{[n,k]}$ follows directly from observing that $\cC_{[n,k]}^\perp \subset \cC_{[n,n/2]}^\perp = \cC_{[n,n/2]} \subset \cC_{[n,k]}$ and the fact that puncturing preserves weak duality.
\end{proof}

\begin{lemma}
  Let $q$ be even with $q\geq n$. For any $[n,k]$ GRS code $\cC$ there exists an $[n,t]$ GRS code $\cD$ such that their star-product $\cS = \cC \star \cD$ is an $[n,k+t-1]$ weakly self-dual GRS code.
  \label{lem:SelfDualStarCode}
\end{lemma}
\begin{proof}
  By \cite{mirandola2015critical} the star product between an $[n,k]$ GRS code $\cC$ with column multipliers $\cV_\cC$ and an $[n,t]$ GRS code $\cD$ with column multipliers $\cV_\cD$, both with the same locators $\cL$, is the $[n,k+t-1]$ GRS code with column multipliers $\cV_\cC\star \cV_\cD$ and code locators $\cL$. Denote by $\cV_\cS$ the column multipliers of a weakly-self dual $[n,k+t-1]$ GRS code with code locators $\cL$, which exists due to Lemma~\ref{lem:existenceSelfDual}. Then, the lemma statement follows from setting $\cV_\cD = (\cV_\cC)^{-1}\star \cV_\cS$, where we denote by $(\cV_\cC)^{-1}$ the element-wise inverse of $\cV_\cC$.
\end{proof}

%\subsection{Quantum Computation}\label{sec:QC}

\begin{figure*}[ht]
\centering
\begin{tikzpicture}[thick,scale=0.9, every node/.style={transform shape}]
\path
(4.35,0) node{
	$\begin{pmatrix}[ccc|ccc]
	    X_{1,1}^{1,1}  & \cdots & X_{1,k}^{1,1} & X_{2,1}^{1,1}  & \cdots & X_{2,k}^{1,1}  \\
	    \vdots & \ddots & \vdots & \vdots & \ddots & \vdots \\
	    X_{1,1}^{1,\b} & \cdots & X_{1,k}^{1,\b} & X_{2,1}^{1,\b} & \cdots & X_{2,k}^{1,\b} \\ \hline
	    \vdots & \vdots & \vdots & \vdots & \vdots & \vdots \\ \hline
	    X_{1,1}^{m,1}  & \cdots & X_{1,k}^{m,1} & X_{2,1}^{m,1}  & \cdots & X_{2,k}^{m,1} \\
	    \vdots & \ddots & \vdots & \vdots & \ddots & \vdots \\
	    X_{1,1}^{m,\b} & \cdots & X_{1,k}^{m,\b} & X_{2,1}^{m,\b} & \cdots & X_{2,k}^{m,\b}
	\end{pmatrix} \quad \cdot G_\cC = \quad
	\begin{pmatrix}[ccc|ccc]
	    Y_{1,1}^{1,1}  & \cdots & Y_{1,n}^{1,1} & Y_{2,1}^{1,1}  & \cdots & Y_{2,n}^{1,1}  \\
	    \vdots & \ddots & \vdots & \vdots & \ddots & \vdots \\
	    Y_{1,1}^{1,\b} & \cdots & Y_{1,n}^{1,\b} & Y_{2,1}^{1,\b} & \cdots & Y_{2,n}^{1,\b} \\ \hline
	    \vdots & \vdots & \vdots & \vdots & \vdots & \vdots \\ \hline
	    Y_{1,1}^{m,1}  & \cdots & Y_{1,n}^{m,1} & Y_{2,1}^{m,1}  & \cdots & Y_{2,n}^{m,1} \\
	    \vdots & \ddots & \vdots & \vdots & \ddots & \vdots \\
	    Y_{1,1}^{m,\b} & \cdots & Y_{1,n}^{m,\b} & Y_{2,1}^{m,\b} & \cdots & Y_{2,n}^{m,\b}
	\end{pmatrix}$
}
(-3.8,1.1) node[blue]{file 1}
(-3.8,-1) node[blue]{file $m$}
(6.1,-2.3) node[orange]{\small $\text{\rmfamily\scshape server}\, 1$}
(7.92,-2.33) node[orange]{\small $\text{\rmfamily\scshape server}\, n$}
(9.38,-2.3) node[orange]{\small $\text{\rmfamily\scshape server}\, 1$}
(11.2,-2.33) node[orange]{\small $\text{\rmfamily\scshape server}\, n$}
;
\draw[thin,blue,rounded corners=4pt] (-2.9,0.25) rectangle (3.15,2.05);
\draw[thin,blue,rounded corners=4pt] (-2.9,-.3) rectangle (3.15,-2);
\draw[thin,orange,rounded corners=4pt] (5.7,-2) rectangle (6.5,2.05);
\draw[thin,orange,rounded corners=4pt] (7.62,-2) rectangle (8.42,2.05);
\draw[thin,orange,rounded corners=4pt] (8.77,-2) rectangle (9.57,2.05);
\draw[thin,orange,rounded corners=4pt] (10.8,-2) rectangle (11.6,2.05);
\end{tikzpicture}
\caption{Illustration of a DSS storing $m$ files, each consisting of $2 \b k$ symbols. The matrix $G_\cC$ is a generator matrix of a $[2n,2k]$ code $\cC$.}
\label{fig:DSS}
\end{figure*}
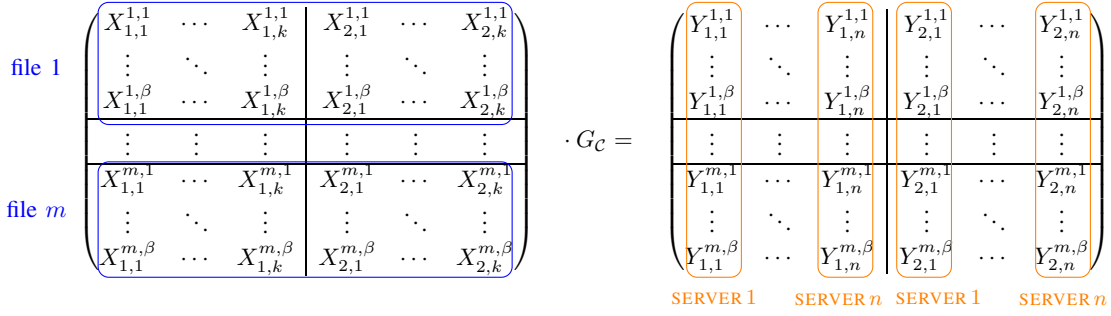 %DSS Figure

\textbf{Quantum Computation.} In this section we collect some notions from non-binary stabilizer formalism~\cite{AK01,KKKS06}. For general notions in quantum computation we refer the reader to~\cite{NC00}.

%\ch{We could comment on the implications of qubit-dimensionality wrt length of the code, code field size, and sub-packetization degree, if there's space. }
Let $q = p^k$ be a prime power and fix $n \in \N$. A quantum system is a $q$-dimensional Hilbert space $\H$ along with a computational basis, that is, a prespecified orthonormal basis $\mathcal{B} = \gpar{\cz,\co,\ldots,|q-1\rangle}$. One typically takes $\H = \C^q$. We will identify the field $\Fq$ with $\F_p^k$ in the usual way. Denote ${\rm tr} : \Fq \rightarrow \F_p, x \mapsto \sum_{i = 0}^{k-1}x^{q^i}$ the corresponding trace function. 
Let $\omega = \exp(2\pi i/p)$ be a $p$-th primitive root of unity. For $a,b\in \Fq$, the maps $\sX(a)|x\rangle = |x+a\rangle \text{ and } \sZ(b)|x\rangle = \omega^{{\rm tr}(bx)}|x\rangle$ are unitary operations on the Hilbert space $\H$. For $c = (c_1,\ldots,c_n)\in \Fq^n$, we extend these maps to unitary transformations of $\C^{q^n}\cong (\C^q)^{\otimes n} = \H^{\otimes n}$ as 
\begin{equation*}
    \sX(c) = \sX(c_1)\otimes\cdots\otimes \sX(c_n) \text{ and } \sZ(c) = \sZ(c_1)\otimes\cdots\otimes \sZ(c_n).
\end{equation*}
A Weyl operator is then defined as $\sW(a,b) = \sX(a)\sZ(b)$, and the Heisenberg-Weyl group ${\rm HW}_{q^n}$ is the subgroup of the unitary group $\mathbb{U}(q^n)$ generated by these operators. A~\emph{stabilizer group} is an abelian subgroup $S \leq {\rm HW}_{q^n}$ such that $-I_{q^n} \notin S$. There is a well-known one-to-one correspondence between stabilizer groups and weakly self-dual subspaces of $\Fq^{2n}$ with respect to the symplectic inner product
\begin{equation}\label{inners}
    \langle (a,b)\mid (c,d) \rangle_J := {\rm tr}\tpar{(a,b) J (c,d)^\top},
\end{equation}
where $J = \ppmatrix{0^{n \times n} & -I_n \\ I_n  & 0^{n \times n}} \in \Fq^{2n \times 2n},\ a,b,c,d \in \Fq^n$. We will denote the dual with respect to~\eqref{inners} of a subspace $\tV\leq \Fq^{2n}$  by $\tV^{\perp_J}$. Based on the above mentioned correspondence, we will identify a stabilizer group $S$ as $S(\tV)$ for some unique $\tV \leq \Fq^{2n}$ with $\tV \subseteq \tV^{\perp_J}$.

Given a stabilizer group $S = S(\tV)$, we have that $\omega^{{\rm tr}(v\cdot w)}$ is an eigenvalue of ${\sf E}(v) \in S$, for any $w \in \tV$, and all its eigenvalues are of this form. 
Let $\mathcal{H}_w$ be the common eigenspace of the operators ${\sf E}(v)$ corresponding to the eigenvalue associated to $w$, and let $P_w: \C^{q^n}\rightarrow \mathcal{H}_w$ be the correspond projector.
It is shown in~\cite[Sec.~III.A]{song2020} that $\cB^\tV = \{P_v \mid v \in \tV\}$ is a projection-valued measurement (PVM), which we will measure with. 
We point out here isomorphishms $\tV \cong {\rm Hom}(\tV,\Fq) \cong \Fq^{2n}/\tV^{\perp_J}$ and for us it will be beneficial to index the projections with cosets $\overline{w} \in \Fq^{2n}/\tV^{\perp_J}$.

%\subsection{Private Information Retrieval}

\textbf{Private Information Retrieval.} Consider a storage system storing $m$ files $X^i, \ i\in[m]$, as described above.

In a PIR protocol a user desiring the $K$-th file $X^K$ chooses a query $Q^K=\{Q_1^K,\ldots,Q_n^K\}$ from a query space $\mathcal{Q}$ and transmits $Q_s^K$ to the $s$-th servers.
In the non-quantum PIR setting the response $A_s^K$ from the $s$-th server is a deterministic function of the received query $Q_s^K$ and the shares of the (encoded) files it stores. We denote by $A^K = \{A_1^K,\ldots, A_n^K\}$ the set of responses from all servers. In this work, we consider an extended setting where the user and the servers are also allowed to communicate quantum systems. Briefly, in this QPIR setting, we have $n$ servers each possessing a $q$-dimensional quantum system. Their composite quantum system is initialized in a specific entangled state. Each server applies some standard quantum operations to its quantum systems (\emph{e.g.}, applying a Weyl operator on a quantum system) depending on (a function of) the received query and the shares of the (encoded) files it stores, and \emph{responds} by sending the remaining quantum systems to the user. The total number of quantum systems that the servers prepare at the beginning of the protocol is denoted by $\q{in}$, while the total number of quantum systems that are transmitted from the servers to the user is denoted by $\q{out}$. In this work, we have $\q{in} = \q{out}$.

%%%%%%%%%%%%%%%%%
%%%%%%%%%%%%%%%%

\begin{definition}[Correctness]\label{def:correctness}
A QPIR protocol is said to be \emph{correct} if the user can retrieve the desired file $X^K, K\in [m]$ from the responses of the servers.
\end{definition}

As usual, we assume honest-but-curious servers who follow the assigned protocol, but might try to determine the index $K$ of the file desired by the user.
\begin{definition}[Privacy with $t$-Collusion]\label{def:privacy}
  \emph{User privacy:} Any set of at most $t$ colluding servers gains no information about the index $K$ of the desired file.\\
  \emph{Server privacy:} The user does not gain any information about the files other than the requested one.\\
  \emph{Symmetric scheme:} A scheme with both user and server privacy is called  \emph{symmetric}.
\end{definition}

Formally, the QPIR rate in this setting is defined in the following. As customary, we assume that the size of the query vectors is negligible compared to the size of the files. This is well justified if the files are assumed to be large, as the upload cost is independent of the size of the files. For simplicity, we only consider files of sizes $2 \b k \log_2(q)$ in the following. However, note that repeatedly applying the scheme with the same queries allows for the download of files that are any multiple of $2 \b k \log_2(q)$ in size at the same rate and without additional upload cost.
\begin{definition}[QPIR Rate]\label{def:QPIRrate}
  For a QPIR scheme, \ie a PIR scheme with classical files, classical queries from user to servers and quantum responses from servers to user, the \emph{rate} is the number of retrieved information bits of the requested file over the binary logarithm of the dimension of the composite quantum system, \emph{i.e.},
  \begin{equation*}
    R_{\mathsf{QPIR}} = 
    \frac{\text{\#information bits in a file}}{\log_2(\dim{\H^{\otimes n}})}.
  \end{equation*}
\end{definition}

For comparison, we also informally define the PIR rate in the non-quantum setting as the number of retrieved information bits of the requested file per downloaded response bit, \emph{i.e.},
\begin{equation*}%\label{eq:PIRrate}
R_{\mathsf{PIR}} = %\frac{H(X^i)}{\sum_{j=1}^n H(A_j^i)} =
\frac{\text{\#information bits in a file}}{\text{\#downloaded bits}}.
\end{equation*}
The PIR capacity is the supremum of PIR rates of all possible PIR schemes, for a fixed parameter setting.

\begin{remark}
\label{rem:Qubit capacity}
In this setting we assume that the user does not share any entanglement with the servers. Hence, the maximal number of information bits obtained when receiving a quantum system, \emph{i.e.}, the number of bits that can be communicated by transmitting a quantum system from a server to the user without privacy considerations, is the binary logarithm of the dimension of the corresponding Hilbert space~\cite{holevo1973bounds}. 
\end{remark}

We would also like to point out that higher-dimensional quantum systems are mainly of theoretical interest. If we restrict to two-dimensional systems while still wishing to protect against collusion, the MDS property should be relaxed in order to allow for binary storage codes. This will likely lower the achievable QPIR rate  but make the scheme otherwise more practical. %, which is also an interesting direction for further study.

\section{$[n,k]$-coded storage with $t$-collusion}

\textbf{Storage.} We consider a storage system as described in Section~\ref{sec:preliminaries} (see Figure~\ref{fig:DSS}). The code $\cC'$ is chosen to be an $[n,k]$ GRS code and for a given integer $c$, which will be defined in the next paragraph, the parameter $\b$ is fixed to $\b = \lcm(c,k)/k$. %\lh{Detailed definitions moved to preliminaries. Only the parameter choices specific to the scheme are here now.}

\textbf{Codes.} 
% Let $q$ be the power of a prime. Let $\cC'$ be an $[n,k]$ GRS code with generator matrix $G_{\cC'}$. We consider $n$ servers containing $m$ files stored according to the code obtained by the Cartesian product $\cC = \cC' \times \cC'$. Hence the generator matrix of $\cC$ is given by $G_\cC = \diag{G_{\cC'},G_{\cC'}} \in \Fq^{2k \times 2n}$.
Let $t$ be the collusion parameter with $\frac{n}{2} \leq k+t-1 < n$. By Lemma~\ref{lem:SelfDualStarCode} there exists an $[n,t]$ GRS code $\cD'$ such that $\cS' = \cC' \star \cD'$ is an $[n,k+t-1]$ weakly self-dual GRS code. We define the query code as the Cartesian product $\cD = \cD' \times \cD'$. Thus, for a generator matrix $G_{\cD'}$ of $\cD'$, the matrix $G_\cD = \diag{G_{\cD'},G_{\cD'}} \in \Fq^{2t \times 2n}$ is a generator matrix of $\cD$.

Define $\cS = \cC \star \cD$ and $\cS' = \cC' \star \cD'$ . By \eqref{eq:starCartesian} we have $\cS = \cC \star \cD = \cS' \times \cS'$, so $\cS$ is the Cartesian product of two star product codes. Define $c=d_{\cS'}-1$, where $d_{\cS'} = n-k-t+2$ is the minimum distance of $\cS'$.
% Denoting $c = \frac{d_\cS - 1}{2} = n - k - t + 1$, we notice that $c = d_{\cS'} - 1$. 

Let $H_{\cS'} \in \Fq^{(n - k - t + 1) \times n}$ be a parity-check matrix of $\cS'$. By Definition~\ref{def:WeaklySelfDualCode}, the code $\cS'$ has a generator matrix of the form $G_{\cS'} = (H_{\cS'}^\top \ \  F_{\cS'}^\top)^\top$ for some $F_{\cS'} \in \Fq^{2(k + t - 1 - n) \times n}$. Hence, $\cS$ has a generator matrix of form
\begin{equation}
    \label{eq:StarGenMatrix}
    G_\cS = \ppmatrix{\diag{H_{\cS'},H_{\cS'}} \\ \diag{F_{\cS'},F_{\cS'}}} \in \Fq^{2(k+t-1) \times 2n}.
\end{equation}

\begin{lemma}
\label{lem:MeasurementMatrix}
Let $G_\cS$ be the matrix defined in Eq.~\eqref{eq:StarGenMatrix} and let $H_\cS$ be the submatrix of $G_\cS$ containing its first $2(n-k-t+1)$ rows. Let $w_1,\ldots,w_{2n}$ be the column vectors of $G_\cS$. Then, they satisfy conditions (a) and (b) of \cite[, Lemma 2]{song2020}, \ie
\begin{itemize}
    \item[(a)] $w_{\pi(1)},\ldots,w_{\pi(k+t-1)},w_{\pi(1) + n},\ldots,w_{\pi(k+t-1) + n}$ are linearly independent for any permutation $\pi \in \mathsf{S}_n$.

    \item[(b)] $H_\cS J^\top G_\cS^\top = 0$.
\end{itemize}
\end{lemma}

\begin{proof}
%\lh{Reformulated.}
It is well-known that any subset of $k+t-1$ columns of the generator matrix of an $[n,k+t-1]$ MDS code are linearly independent. Hence, the columns $w_{\pi(1)},\ldots,w_{\pi(k+t-1)}$ are linearly independent, as the first $n$ columns of $G_\cS$ generate $\cS$. The same holds for $w_{\pi(1) + n},\ldots,w_{\pi(k+t-1) + n}$. Trivially, any non-zero columns of a diagonal matrix are linearly independent and property (a) follows.% In particular, the first $n$ columns (without the zeros) correspond to the generator matrix of the GRS code $\cS'$, which is MDS. Thus, property (a) is easily satisfied by the properties of MDS codes.
% It is known the columns of the generator matrix of an MDS code are linearly independent for any permutation \lh{Not all, only $k$.}. In particular, the first $n$ columns (without the zeros) correspond to the generator matrix of the GRS code $\cS'$, which is MDS. Thus, property (a) is easily satisfied by the properties of MDS codes.

Property (b) follows directly from observing that, by definition, $HG^\top = 0$ for any linear code with generator matrix $G$ and parity-check matrix $H$.
% Property (b) is easily satisfied by the well known equality $HG^\top = 0$ for any linear code with generator matrix $G$ and parity-check matrix $H$.
\todoLater{Now, we want to prove property (b). Since for any linear code $\mathcal{L}$ with generator matrix $G$ and parity-check matrix $H$ it holds that $HG^\top = 0$, then
\[ \begin{split}
    H_\cS J^\top G_\cS^\top & =
    \ppmatrix{0 & H_{\cS'} \\ -H_{\cS'} & 0}
    \ppmatrix{G_{\cS'}^\top & 0 \\ 0 & G_{\cS'}^\top} \\
    & = \ppmatrix{0 & H_{\cS'} G_{\cS'}^\top \\ -H_{\cS'} G_{\cS'}^\top & 0} = 0.
\end{split} \]}
\end{proof}

\textbf{Targeting servers.} 
Suppose the desired file is $X^K$. We define the indexing such that the file can be obtained in $\rho = \lcm(c,k)/c$ rounds. During each of these rounds, the user can download $2c/\b=2k/\rho$ symbols from each of the $\b$ rows of $Y^K$, where the factor 2 is achieved by utilizing the properties of superdense coding~\cite{NC00} in quantum computation. 

Fix $\cJ = \gpar{1,\ldots,\max \gpar{c,k}}$ to be the set of server indices from which the user obtains the symbols of $Y^K$. We define $\cJ_r^b\subseteq \cJ$ with $|\cJ_r^b| = c/\b$
% =  \gpar{\cJ_{r,1}^b,\ldots,\cJ_{r,c}^b} \subseteq \cJ$ 
as in \cite[Eq. (22)]{freij2017private}, where $r \in [\rho]$ and $b \in [\b]$, and denote $\cJ_r = \cJ_r^1 \cup \ldots \cup \cJ_r^\b$. This definition ensures that during the $r$-th iteration the user obtains the symbols $(Y_{1,a}^{K,b},Y_{2,a}^{K,b})$ for every $a \in \cJ_r^b$ and $b\in [\b]$.

% We now build matrices $N^{(1)},\ldots,N^{(\rho)}$ such that the row vectors of the matrix $(G_{\cS_1}^\top \ \ (N^{(r)})^\top)^\top$ form a basis for $\Fq^n$ for any $r \in [\rho]$. In the following, we fix the iteration $r \in [\rho]$.
% Let $\cJ_r^b = \gpar{\cJ_{r,1}^b,\ldots,\cJ_{r,c}^b}$ be defined as in \cite[Eq. (27)]{freij2017private} and let $\cJ_r = \cJ_r^1 \cup \ldots \cup \cJ_r^\b$ be the set of servers we are targeting. 
We define
\begin{equation}
    \label{eq:MatrixN}
    % N^{(r)} = \ppmatrix{e_{\cJ_{r,1}}^{n} & \cdots & e_{\cJ_{r,c}}^{n}}^\top \in \Fq^{c \times n}.
 N^{(r)} = \ppmatrix{e_{a}^{n}}_{a \in \cJ_{r}}^\top \in \Fq^{c \times n}.
\end{equation}
Then, the matrix $(G_\cS^\top \ \ (M^{(r)})^\top)^\top$, with $M^{(r)} = \diag{N^{(r)},N^{(r)}} \in \Fq^{2c \times 2n}$, is a basis for $\Fq^{2n}$. To see that this is in fact a basis observe that the row span of $N^{(r)}$, by definition, contains vectors of weight at most $c$. The span of $G_{\cS'}$ contains vectors of weight at least $ d_{\cS'}=c+1$. It follows that the spans of $N^{(r)}$ and $G_{\cS'}$ intersect trivially, which implies that their ranks add up.

\todoLater{
\lh{Reformulated, see above.}
Suppose the desired file is $X^K$. During each round of the scheme, we can download $2g = 2 k/ \rho = 2 c / \b$ symbols from every row of $Y^K$, where the factor 2 is essentially obtained by the properties of superdense coding~\cite{NC00} in quantum computation. After $\rho$ rounds, the scheme will have downloaded $2 \rho g = 2k$ \lh{What is $g$?} \ma{It's defined in the previous sentence :)} symbols of the $b$-th row $Y^{K,b}$ of $Y^K$ for all $b \in [\b]$. 

We fix a subset $\cJ = \gpar{1,\ldots,\max \gpar{c,k}} \subseteq [n]$ of servers which stays constant throughout the scheme. The set $\cJ$ will be the set of all servers from which we obtain encoded symbols. We will also make use of sets $\cJ_r^b \subseteq \cJ$~\cite[Eq. (22)]{freij2017private} where $r \in [\rho]$ and $b \in [\b]$, which are defined so that during the $r$-th iteration we obtain the symbols $(Y_{1,s}^{K,b},Y_{2,s}^{K,b})$ from every server $s \in \cJ_r^b$.

We now build matrices $N^{(1)},\ldots,N^{(\rho)}$ such that the row vectors of the matrix $(G_{\cS_1}^\top \ \ (N^{(r)})^\top)^\top$ form a basis for $\Fq^n$ for any $r \in [\rho]$. In the following, we fix the iteration $r \in [\rho]$. Let $\cJ_r^b = \gpar{\cJ_{r,1}^b,\ldots,\cJ_{r,c}^b}$ be defined as in \cite[Eq. (27)]{freij2017private} and let $\cJ_r = \cJ_r^1 \cup \ldots \cup \cJ_r^\b$ be the set of servers we are targeting. We define
\begin{equation}
    \label{eq:MatrixN}
    N^{(r)} = \ppmatrix{e_{\cJ_{r,1}}^{n} & \ldots & e_{\cJ_{r,c}}^{n}}^\top \in \Fq^{c \times n}.
\end{equation}
Hence, the matrix $M^{(r)} = \diag{N^{(r)},N^{(r)}} \in \Fq^{2c \times 2n}$ is such that the row vectors of $(G_\cS^\top \ \ (M^{(r)})^\top)^\top$ form a basis for $\Fq^{2n}$. To see that this is in fact a basis observe that the row span of $N^{(r)}$, by definition, contains vectors of weight at most $c$ and, since $d_{\cS'}=c+1$, the span of $G_{\cS'}$ contains vectors of weight at least $c+1$. It follows that the spans of $N^{(r)}$ and $G_{\cS'}$ intersect trivially, which implies that their ranks add up.}

\subsection{A coded QPIR scheme} 
\label{sec:scheme}

Let $\tV$ be the space spanned by the first $2(n-k-t+1)$ rows of $G_\cS$ and $\Fq^{2n} / \tV^{\perp_J} = \gpar{\overline{w} = w + \tV^{\perp_J} : w \in \langle M^{(r)}\rangle_{\mathsf{row}}}$, where $\langle M^{(r)}\rangle_{\mathsf{row}}$ is the space spanned by the rows of $M^{(r)}$. By Lemma~\ref{lem:MeasurementMatrix}, the rows of $G_\cS$ span the space $\tV^{\perp_J}$. 

We now describe the five steps of our QPIR scheme. The first four steps are repeated in each round $r \in [\rho]$.

\textbf{Distribution of entangled state.} Let $\H_1,\ldots,\H_n$ be $q$-dimensional quantum systems and $\smix = q^{n-2(k+t-1)} \cdot I_{q^{2(k+t-1)-n}}$. %\lh{We use $\rho$ and $\rho_\text{mix}$, but they aren't related.} \ma{changed to $\smix$}
% $\smix = \tpar{\frac{1}{q^{2(k+t-1)-n}} \cdot I_{q^{2(k+t-1)-n}}}$.
By \cite[Eq.~(18)]{song2020} the composite quantum system $\H = \H_1 \otimes \dots \otimes \H_n$ is decomposed as $\H = \cW \otimes \C^{q^{2(k+t-1)-n}}$, where $\cW = \Span{\cb{\overline{w}} \mid \overline{w} \in \Fq^{2n} / \tV^{\perp_J}}$. The state of $\H$ is initialized as $\cb{\overline{0}} \rb{\overline{0}} \otimes \smix$ and distributed such that server $s \in [n]$ obtains $\H_s$.

\textbf{Query.} The user chooses a matrix $Z^{(r)} \in \Fq^{m\b \times 2t}$ uniformly at random. %The row with coordinate $(i,b)$ is denoted by $Z^{i,b,(r)}$.
%\lh{Above replaces from here:} The user samples uniformly at random $Z^{i,(r)} \in \Fq^{\b \times 2t}$ for each $i \in [m]$. We denote each sample as $Z^{i,(r)} = \begin{pmatrix}[c|c]
%    Z_{1,\tau}^{i,b,(r)} & Z_{2,\tau}^{i,b,(r)}
%\end{pmatrix}$,
%where $Z_{p,\tau}^{i,b,(r)} \in \Fq, \quad p \in [2], \tau \in [t], i \in [m], b \in [\b]$. \lh{It doesn't look like we need to index these $z$ individually as symbols. Can't we define the whole thing by saying we choose $r$ a random matrices $Z^{(r)}$?} \ma{True, it's not needed :)}
%In particular, $Z^{i,b,(r)}$ is the $b$-th row of $Z^i$. We denote by $Z^{(r)} =
%\begin{pmatrix}[c|c]
%    Z_{1,\tau}^{(r)} & Z_{2,\tau}^{(r)}
%\end{pmatrix} \in \Fq^{m\b \times 2t}$ the matrix of all the random symbols, where each column $Z_{p,\tau}^{(r)}$ contains the symbol $Z^{i,b,(r)}_{p,\tau}$ in the row with coordinate $(i,b)$. \lh{to here.}
We define $E_{(K)} \in \Fq^{m\b \times 2c}$ with $E_{(K),p,a} = e_{(K,a)}^{m\b},\ p\in[2], a\in[c]$. %Let $E_K =
% \begin{pmatrix}[c|c]
%\big(e_{(K,a)}^{m\b} \ | \  e_{(K,a)}^{m\b}\big)%\end{pmatrix}
%\in \Fq^{m\b \times 2c}$. \lh{This is an $m\b \times 2$ matrix. I'm confused, is there some running index missing? Either way, how about: "We define $E_K \in \Fq^{m\b \times 2c}$ with $(E_K)_{p,a} = e_{(K,a)}^{m\b}, p\in[2], a\in[c]$."  Then we can again just treat $E$ as a matrix.}
Notice that the row in coordinate $(i,b)$ of the product $E_{(K)} \cdot M^{(r)}$ is $\sum_{p=1}^2 \sum_{a \in \cJ_{r}^b} \d_{i,K} (e_{(p,a)}^{2n})^\top$. We denote by $Q^{(r)} \in \Fq^{m\b \times 2n}$ the matrix of all the queries, which are computed as
\begin{equation}
    Q^{(r)} = (Z^{(r)} \ \ E_{(K)}) \cdot \ppmatrix{G_\cD \\ M^{(r)}} = Z^{(r)} \cdot G_\cD + E_{(K)} \cdot M^{(r)}.
    \label{eq:Queries}
\end{equation}
Each server $s \in [n]$ receives two vectors $Q^{(r)}_{1,s},Q^{(r)}_{2,s} \in \Fq^{m\b}$. %For each subset of $t$ servers, the corresponding joint distribution of queries is the uniform distribution over $\Fq^{m\b \times 2t}$.

\textbf{Response.} The servers compute the dot product of each column of their stored symbols and the respective column of the queries received, \ie they compute the response $A^{(r)}_{p,s} = Y_{p,s}^\top \cdot Q^{(r)}_{p,s},\ s \in [n], p \in [2]$. Each $A^{(r)}_{p,s}$ is a symbol in $\Fq$. Server $s$ applies $\sX(A^{(r)}_{1,s})$ and $\sZ(A^{(r)}_{2,s})$ to its quantum system and sends it to the user.

\textbf{Measurement.} The user applies the PVM $\cB^\tV = \gpar{P_{\overline{w}} \mid \overline{w} \in \Fq^{2n} / \tV^{\perp_J}}$ on $\H$ and obtains the output $o^{(r)} \in \Fq^{2c}$. %$\begin{pmatrix}[c|c] Y_{1,\cJ_{r,a}}^{K,a} & Y_{2,\cJ_{r,a}}^{K,a} \end{pmatrix} \in \Fq^{2c}$. 

\textbf{Retrieval.} Finally, after $\rho$ rounds the user has retrieved $2\rho c = 2 \b k$ symbols of $\Fq$ from which he can recover the desired file $X^K$. 
%the symbols $\begin{pmatrix}[c|c] Y_{1,\k}^{K,b} & Y_{2,\k}^{K,b} \end{pmatrix} \in \Fq^{\b \times 2k}$ and can easily recover the desired file $X^K$ by solving a system of linear equations. 

\subsection{Properties of the coded QPIR scheme} \label{sec:Properties of coded QPIR}

\begin{lemma}
The scheme of Section~\ref{sec:scheme} is correct, \emph{i.e.}, fulfills Definition~\ref{def:correctness}.
\end{lemma}
\begin{proof}
Let us fix $r \in [\rho]$. By \cite[Lemma 1]{song2020} the state after the servers' encoding is
\[
\sW(A^{(r)}) (\cb{\overline{0}} \rb{\overline{0}} \otimes \smix) \sW(A^{(r)})^\dagger = \cb{\overline{A^{(r)}}} \rb{\overline{A^{(r)}}} \otimes \smix.  
\]

We observe that $\tV^{\perp_J} = \cS$ since both spaces are spanned by the rows of $G_\cS$. By definition of the star product scheme, the response vector is
\begin{equation}
    \label{eq:Response}
    \begin{split}
        A^{(r)} = & \;
        \begin{pmatrix}[c|c]
            A^{(r)}_{1} & A^{(r)}_{2}
        \end{pmatrix} = \sum_{i=1}^m \sum_{b=1}^{\b} Y^{i,b} \star Q^{(r),i,b} \\
        = & \; \sum_{i=1}^m \sum_{b=1}^{\b} \tpar{X^{i,b} \cdot G_\cC} \star \tpar{Z^{(r),i,b} \cdot G_\cD} \\
        & + \sum_{i=1}^m \sum_{b=1}^{\b} Y^{i,b} \star \Big(\sum_{a \in \cJ_r^b} \d_{i,K} \big(e_{(1,a)}^{2n} + e_{(2,a)}^{2n} \big)^\top \Big)\\
        & \; \in \cS + \sum_{b=1}^\b \sum_{a \in \cJ_r^b} \big(Y_{1,a}^{K,b} e_{(1,a)}^{2n} + Y_{2,a}^{K,b} e_{(2,a)}^{2n}\big)^\top \\
        & \hspace{10pt} = \tV^{\perp_J} +
        \begin{pmatrix}[c|c]
            Y_{1,a}^{K,b} & Y_{2,a}^{K,b}
        \end{pmatrix}_{a \in \cJ_r^b, b \in [\b]} \cdot M^{(r)}.
    \end{split}
\end{equation}
The random part is encoded into a vector in $\tV^{\perp_J}$ while the vector $\big( Y_{1,a}^{K,b} \ | \ Y_{2,a}^{K,b} \big)_{a \in \cJ_r^b, b \in [\b]} \in \Fq^{2c}$ is encoded with $M^{(r)}$ and hence independent of the representative of $\overline{o^{(r)}}$. Therefore, the user obtains the latter without error after measuring the quantum systems with the PVM $\mathcal{B}^\tV$. After $\rho$ rounds the user retrieved the symbols $\big( Y_{1,\k}^{K,b} \ | \ Y_{2,\k}^{K,b} \big)_{\k \in [k]} \in \Fq^{2k}$ for each $b \in [\b]$ and can recover the desired file $X^K$ by solving a system of linear equations.
\end{proof}

\begin{lemma} \label{lem:Secrecy}
The scheme of Section~\ref{sec:scheme} is symmetric and protects against  $t$-collusion in the sense of Definition~\ref{def:privacy}.
\end{lemma}
\begin{proof}
Privacy in the quantum part of the protocol follows directly from the privacy of the protocol with all but one servers colluding. For details, we refer the reader to \cite{song2020}. %Probably we can eliminate this -> Intuitively, a colluding set of servers has access to their received queries and quantum states entangled with states of servers (possibly) outside of the colluding set. They could infer some more information about the file index by measuring these states, however, as this would consume the entanglement, the correctness of the protocol cannot be guaranteed anymore in this case. As the servers are honest-but-curious, \ie have to ensure that the user obtains the requested file, this is not possible with the specified amount of entanglement. On the other hand, any set of at most $t$ colluding nodes can in any case share all data between them, so increasing the number of quantum systems sharing an entangled state within this colluding set does not breach privacy. It remains to be shown that the queries received by a set of at most $t$ colluding nodes reveal no information about the file index, which directly follows from the privacy of the protocol in \cite{freij2017private}. The idea is that,
User privacy is achieved since, for each subset of $t$ servers, the corresponding joint distribution of queries is the uniform distribution over $\Fq^{m\b \times 2t}$. %For completeness, we include a short proof here for both user and server privacy. Consider a set $\cT \subset [n]$ with $|\cT|\leq t$ of colluding servers. The set of queries these servers receive is given by $Q^{(r)}$ during round $r \in [\rho]$. By the MDS property of the code $\cD$ any subset of $t$ columns of $G_\cD$ is linearly independent. As the columns of $Z^{(r)}$ are uniformly distributed and chosen independently for each $r \in [\rho]$, any subset of $t$ columns of $Z^{(r)} \cdot G_\cD$ is statistically independent and uniformly distributed. The sum of a uniformly distributed vector and an independently chosen vector is again uniformly distributed, and therefore adding the matrix $E_K \cdot M^{(r)}$ does not incur any dependence between any subset of $t$ columns and the file index $K$.
For each $r \in [\rho]$, server secrecy is achieved because in every round the received state of the user is independent of $Y^i$ with $i \neq K$.
\end{proof}

Unlike in the classical setting, the servers in the quantum setting do not need access to a source of shared randomness that is hidden from the user to achieve server secrecy. However, this should not be viewed as an inherent advantage since the servers instead share entanglement.

\begin{theorem} \label{thm:rate}
The QPIR rate of the scheme in Section~\ref{sec:scheme} is
\begin{equation*}
    R_{\mathsf{QPIR}} = \frac{2(n - k - t + 1)}{n}
\end{equation*}
\end{theorem}
\begin{proof}
The user downloads $\rho n$ quantum systems while retrieving $2 k \b \log_2(q)$ bits of information, thus the rate is given by 
\begin{align*}
R_\mathsf{QPIR} &= \frac{2 k \b \log_2(q)}{\log_2(q^{\rho n})} \\
&= \frac{2 \rho c \log_2(q)}{\rho n \log_2(q)} = \frac{2(n - k - t + 1)}{n}.
\end{align*}
\end{proof}

\begin{remark}
If the collusion parameter  $t$ is such that $1 \leq k+t-1 < n/2$, the presented scheme for $t = n/2 - k + 1$ for even $n$ has rate 1. Since the rate cannot be greater than 1, it is capacity achieving. If $n$ is odd, we just consider $n-1$ servers and $t=(n+1)/2 - k$ in order to achieve rate 1.
\end{remark}

\section{$[6,3]$-coded storage example with $2$-collusion}

Let us choose $q = 7$, $n=6$ and $k=3$. We consider a $[6,3]$ primitive Reed-Solomon (PRS) code~\cite[Ch.~10.2]{macwilliams1977theory} $\cC'$ with generator matrix
\[
G_{\cC'} = \ppmatrix{1 & 1 & 1 & 1 & 1 & 1 \\
                     1 & 3 & 2 & 6 & 4 & 5 \\
                     1 & 2 & 4 & 1 & 2 & 4}.
\]
We have 6 servers containing $m$ files stored according to the Cartesian product $\cC = \cC' \times \cC'$ with generator matrix $G_\cC = \diag{G_{\cC'},G_{\cC'}} \in \F_7^{6 \times 12}$. Let $t=2$ and let $\cD'$ be a $[6,2]$ PRS code $\cD'$ with generator matrix
\[
G_{\cD'} = \ppmatrix{1 & 1 & 1 & 1 & 1 & 1 \\
                     1 & 3 & 2 & 6 & 4 & 5}.
\]
The query code is the Cartesian product $\cD = \cD' \times \cD'$ with generator matrix $G_\cD = \diag{G_{\cD'},G_{\cD'}} \in \F_7^{4 \times 12}$.

The star product code $\cS = \cC \star \cD$ has distance $d_\cS = 3$. Thus, from each server the user can download at most $c = 2$ blocks of information per round. By Eq.~\eqref{eq:starCartesian}, since both $\cC$ and $\cD$ are Cartesian products of PRS codes, also $\cS$ is the Cartesian product of two PRS codes generated by $\cS' = \cC' \star \cD'$. Let
\[
    G_{\cS'} = \ppmatrix{1 & 3 & 2 & 6 & 4 & 5 \\
                         1 & 2 & 4 & 1 & 2 & 4 \\
                         1 & 1 & 1 & 1 & 1 & 1 \\
                         1 & 6 & 1 & 6 & 1 & 6} =
                \ppmatrix{H_{\cS'} \\ F_{\cS'}} \in \F_7^{4 \times 6}
\]
be the generator matrix of the star product code $\cS'$, where $H_{\cS'} \in \F_7^{2 \times 6}$ is the standard parity-check matrix of $\cS'$ and $F_{\cS'} \in \F_7^{2 \times 6}$. One can check that $\cS'$ is indeed a weakly self-dual PRS code. Then the generator matrix of $\cS$ is given by 
\[
G_\cS = \ppmatrix{\diag{H_{\cS'},H_{\cS'}} \\ \diag{F_{\cS'},F_{\cS'}}} \in \F_7^{8 \times 12}.
\]

Each file is divided into pairs of $k=3$ pieces and $\b=2$ blocks. The user will need a total of $\rho = 3$ rounds in order to download the necessary information and reconstruct the desired file. Each server contains a matrix of symbols in $\F_7^{2m \times 2}$. For example, server 2 stores $$Y_{p,2}^{i,b} = X_{p,1}^{i,b} + 3 X_{p,2}^{i,b} + 2 X_{p,3}^{i,b}$$ for $i \in [m], b \in [2], p \in [2]$. %More generally, server $s \in [n]$ will store the pair $(Y_{1,s}^{i,b},Y_{2,s}^{i,b})$ in the row corresponding to coordinate $(i,b)$, $i \in [m], b \in [\b]$.

We fix $\cJ = [3]$, so $\cJ_1 = [2]$ and $\cJ_1^1 = \gpar{1}, \cJ_1^2 = \gpar{2}$. Thus, according to Eq.~\eqref{eq:MatrixN}, we set $N^{(1)} = (I_2 \ \ 0^{2\times 4})$.
%\[
%N^{(1)} = \ppmatrix{1 & 0 & 0 & 0 & 0 & 0 \\
%                    0 & 1 & 0 & 0 & 0 & 0}.
%\]
Hence, $M^{(1)} = \diag{N^{(1)},N^{(1)}} \in \F_7^{4 \times 12}$ is such that the row vectors of the matrix $(G_\cS^\top \ \ (M^{(1)})^\top)^\top$ form a basis for $\F_7^{12}$.

First, the quantum systems are prepared and distributed to the servers according to the first step of the scheme. 

The user samples uniformly at random $Z^{(1)} \in \F^{2m \times 4}_7$. Let $$E_{(K)} = \ppmatrix{e_{(K,1)}^{2m} & e_{(K,2)}^{2m} & e_{(K,1)}^{2m} & e_{(K,2)}^{2m}} \in \F_7^{2m \times 4}.$$ Notice that the row in coordinate $(i,b)$ of the product $E_{(K)} \cdot M^{(1)}$ is
\[
    \d_{i,K} \Big( \d_{b,1} \big( e_{(1,1)}^{12} + e_{(2,1)}^{12} \big) + \d_{b,2} \big( e_{(1,2)}^{12} + e_{(2,2)}^{12} \big) \Big)^\top.
\]
Then, with this choice, the user will retrieve the first block (with $\d_{b,1}$) of the symbols stored in server 1 (with $e_{(p,1)}^{12}$) and the second block (with $\d_{b,2}$) of the symbols stored in server 2 (with $e_{(p,2)}^{12}$) with the desired position $K$ (with $\d_{i,K}$). The user generates the queries according to Eq.~\eqref{eq:Queries} and sends them to the servers. For example, the query to server 2 has symbols $$Q_{p,2}^{(1),i,b} = Z_{p,1}^{(1),i,b} + 3 Z_{p,2}^{(1),i,b} + \d_{i,K} \d_{b,2}$$ for $i \in [m], b \in [2], p \in [2]$.  %For each pair of servers, the corresponding joint distribution of queries is the uniform distribution over $\F_7^{2m \times 4}$.

The servers compute the responses $A^{(1)}_{p,s} = Y_{p,s}^\top \cdot Q_{p,s}^{(1)} \in \F_7$, $p \in [2], s \in [6]$. Server $s$ applies $\sX(A^{(1)}_{1,s})$ and $\sZ(A^{(1)}_{2,s})$ to its quantum system and sends it to the user.

By Eq.~\eqref{eq:Response}, the response vector is
\[ \begin{split}
    A^{(1)} \in \tV^{\perp_J} + \ppmatrix{Y_{1,1}^{K,1} & Y_{1,2}^{K,2} & Y_{2,1}^{K,1} & Y_{2,2}^{K,2}} \cdot M^{(1)}.
\end{split} \]
Then, the user obtains $\tpar{Y_{1,1}^{K,1}, Y_{1,2}^{K,2}, Y_{2,1}^{K,1}, Y_{2,2}^{K,2}} \in \F_7^4$ as output without error.

The other two rounds are analogous by choosing $\cJ_2^1 = \gpar{2}, \cJ_2^2 = \gpar{3}$ and $\cJ_3^1 = \gpar{3}, \cJ_3^2 = \gpar{1}$.

Finally, after 3 rounds the user recovers the symbols $(Y_{1,\k}^{K,b} \ | \ Y_{2,\k}^{K,b}) \in \F_7^{2}$ for each $b \in [2], \kappa \in [3]$. From these symbols the user can easily recover the desired file $X^K$ by solving a system of linear equations. The user downloaded a total of 18 $7$-dimensional quantum systems and gathered 12 symbols of $\F_7$, thus the rate is given by $R_\mathsf{QPIR} = \frac{12}{18} = \frac{2}{3}$.

%\section{Conclusions} 

%Future work consists of adding Byzantine and non-responsive servers to the scheme and proving a converse bound for the capacity. Furthermore, an extension to non-MDS codes and in particular transitive codes \cite{freij2019tpir} could be considered. We would also like to point out that higher-dimensional quantum systems are mainly of theoretical interest. If we restrict to two-dimensional systems while still wishing to protect against collusion, we should relax on the MDS-property to be able to consider binary codes. This will likely lower the achievable QPIR rate, while making the scheme otherwise more practical, which is also an interesting direction for further study.

\section*{Acknowledgments}
The authors would like to thank Prof.~M. Hayashi and S.~Song for helpful discussions.
%%%%%%%%%

\newpage

%\ch{Song references need updating (give the conference name, ITW/ISIT and the arxiv).}\lh{Done.} \ch{Wang-Skoglund: add or replace by the journal versions.} \lh{Done.}

\bibliographystyle{IEEEtran}
\bibliography{main}

\end{document}